\documentclass[11pt,a4paper]{amsart}
\usepackage{amssymb}
\usepackage{graphicx,latexsym,amsmath}

\theoremstyle{plain}
\newtheorem{theorem}{Theorem}[section]
\newtheorem{corollary}[theorem]{Corollary}
\newtheorem{lemma}[theorem]{Lemma}
\newtheorem{proposition}[theorem]{Proposition}

\theoremstyle{definition}

\theoremstyle{remark}
\newtheorem{remark}{Remark}[section]

\numberwithin{equation}{section}

\numberwithin{table}{section}

\numberwithin{figure}{section}

\begin{document}

\title[Stability of Linear Flocks on a Ring Road]{Stability of Linear Flocks on a Ring Road}

\author{J.J.P. Veerman}
\address{Department of Mathematics \& Statistics, Portland State
University, Portland, OR 97201, USA} \email{veerman@pdx.edu}

\author{C.M. da Fonseca}
\address{CMUC, Department of Mathematics, University of Coimbra,
3001-454 Coimbra, Portugal} \email{cmf@mat.uc.pt}
\thanks{This work was supported by CMUC - Centro de Matem\'atica da
Universidade de Coimbra.\\
JJPV enjoyed the gracious hospitality of the Univeristy of Coimbra
while part of this work was done.}

\subjclass[2000]{34D05; 93B52; 92D50}

\date{July 2009}

\keywords{Dynamics of flocks; Stability}

\begin{abstract}
We discuss some stability problems when each agent of a linear flock in $\mathbb{R}$ interacts with its
two nearest neighbors (one on either side).

\end{abstract}

\maketitle

\section{Introduction}

 A \emph{flock} consists of a large number of moving physical objects, called \emph{agents}, with their positions being controlled in such a way that they move along a prescribed path, in a prescribed and fixed configuration (or \emph{in formation}). Each agent knows the position and velocity of only a few other agents, and this flux of information defines a \emph{communication graph}.

In \cite{flocks2} graph theory and linear systems techniques are combined to provide a framework for studying the control of formations.  The main tools of graph theory that are related to the problem are the directed graph Laplacians and the connectedness of the graph. Linear feedback is then used to stabilize the patterns.

More recently \cite{flocks4} studied a system of coupled linear differential equations describing the movement of cars in $\mathbb{R}$, where each car reacts only to its immediate neighbors, and only the movement of the first agent (the \emph{leader}) is independent from the rest of the group. In the paper it was
proved that When equal attention is paid to both neighbors perturbations in the orbit of the leader grow as they propagate through the flock. In fact, perturbations grow proportional to size of the flock: when the leader's perturbation has amplitude 1, then the perturbation in the orbit of the
agent furthest away from the leader will be proportional to $N$ (the size of the flock).

The aim of this paper is first of all to study the (asymptotic) stability of a family 
of such systems. This is answered in detail in Theorems \ref{theo:stable}
and \ref{theo:stable2}. Next we assume the flock has a leader that chooses its
orbit independently of the other members of the flock. We then analyze how exactly
the system converges to a stable flight pattern if the leader changes its orbit.
This question only makes sense when the system is already asymptotically stable, which is therefore
assumed henceforth. The latter question is important in applications as too great fluctuations
in the course of convergence to a coherent flight pattern will make that flock unviable.

In all of these arguments we closely follow the reasoning set forth in previous works \cite{flocks5,flocks6,flocks7}. However there are two important differences. The first is that
the farthest member of the flock (in this work) is coupled to the leader. In the language
of partial differential equations, this is akin to changing a boundary condition. The reason is twofold.
Changing the boundary condition can greatly aid the mathematics, and therefore help
to gain insight. The second reason is a deeper one: we do not know how these boundary
condition influence the stability of these systems, and thus this note can be viewed
a test case (when compared with the papers just cited). The other difference with
the previous papers is that we here allow the weight of the coupling with neighbors
to be \emph{negative}. While at first glance this seems a little odd, there is a good
reason to do so, if one hopes to study systems with more than nearest neighbor coupling.
Suppose for example that one models local interaction as a discretization of a fourth
derivative, a very natural idea. However the couplings to the first and second nearest
neighbors will now have different signs. This goes against the grain of what one
knows about Laplacian systems in general, where in general all couplings must have the same sign
(see \cite{laplacians}). In this case we managed to overcome that problem and analyze stability also when the signs of the (nearest
neighbor) interactions are different. (By necessity they must add up to 1.)

The outline of this paper is as follows. In the next section we start by specifying the model.
Next we discuss the asymptotic stability of the model.
Following \cite{flocks6} we introduce two other types of stability for flocks.
These describe the effect of perturbations in the leaders motion
on the outlying members of the flocks. A flock with $N$ agents is harmonically stable
if the effect of a harmonic motion of the leader on the outlying members
grows less than exponentially fast in $N$ (everything else held fixed).
A flock is said to be impulse stable if the effect of the leader
being kicked is less than exponential on the outlying members. Thus is section
4, we discuss harmonic stability of the model. The problem of impulse stability
is still unsolved. We present a few comments on that problem in section 5.
(The appendix contains technical results and is included for completeness.)

\section{The model} \label{chap:model}

We begin this section establishing the model of this work. The $N+1$ agents move
in $\mathbb{R}$ along orbits $x_i(t)$, $i\in\{0,\cdots n\}$, with velocities
$x_i'(t)$. When they are moving in the desired \emph{formation} their velocities
are equal and their relative positions are determined by $N+1$ a priori given constants $h_i$:
\begin{equation}
x_j-x_{i}= h_j-h_{i} \quad .
\end{equation}
We write the equation of motion for this model in terms of
\begin{equation}
z_i\equiv x_i-h_i \quad .
\end{equation}
These then have the following form:
\begin{equation}\label{example2a}
    \ddot z_i = f\left\{z_{i}- (1-\rho)z_{i-1}-\rho z_{i+1}\right\}
+g\left\{\dot z_i-(1-\rho)\dot z_{i-1}-\rho \dot z_{i+1}\right\}\, ,
\end{equation}
for all $i=1,\ldots, N$, and
\begin{equation}\label{example2b}
    z_{N+1}(t) = z_0(t)\, ,
\end{equation}
\emph{a priori} given. We will assume the feedback parameters $f$, $g$ are negative reals
and the weight $\rho$ is a arbitrary real number.

It is intuitively convenient, though not necessary, to keep a particular realization
of the above system in mind. Identify $x=N+1$ with $x=0$, so that the agents
move on a (topological) circle. Suppose further that the offsets $h_i$ are given by
$h_i=-i \mod N$. Now the desired configuration is that of $N+1$ agents moving
at constant speed and uniformly distributed along a circle. (This explains our
title.)

Our strategy here is primarily studying qualitative aspects of the solution of (\ref{example2a})-(\ref{example2b}) as we let $N$ tend to infinity while keeping
 all other parameters ($\rho$, $f$, and $g$) fixed.
In particular we wish to understand (1) when the system is asymptotically stable
and (2) how does it converge to its equilibrium when it is asymptotically stable.
This stable equilibrium is given by the
two parameter family of orbits:
$$z_k(t)=z_0(0)+v_0(0)\, t$$ and $$\dot z_k(t)= v_0(0) \quad .$$
 These orbits are called \emph{in formation orbits} (for a more detailed discussion, cf. \cite{flocks4, flocks5, flocks6, flocks7}).

It is advantageous to write (\ref{example2a})-(\ref{example2b}) in a more compact form:
$$z\equiv (z_1,\dot z_1, z_2,\dot z_2,\cdots, z_N,\dot z_N)\, .$$
The system can now be recast as a first order ordinary differential equation:
\begin{equation}
\dot z = M z + \Gamma_0(t)\, .
\label{eq:indepleader}
\end{equation}
The matrix $M$ and the vector $\Gamma_0$ are defined below.

Setting
\begin{equation}\label{Q_rho}
Q_\rho=\left(\begin{array}{ccccc}
0 & \rho & & & \\
1-\rho & 0 & \rho & & \\
 & \ddots & \ddots &  \ddots & \\
 & & 1-\rho & 0 & \rho \\
 & & & 1-\rho & 0
\end{array}\right)_{N\times N},
\end{equation}
the matrix $P$ defined by
\begin{equation}
P= I-Q_\rho\, ,
\label{eq:laplacian}
\end{equation}
where $I$ is the $N$-dimensional identity matrix, is called the \emph{reduced graph
Laplacian}. It describes the flow of information among the agents, with the exception of the
leader (hence the word `reduced').

The orbit of the leader is assumed to be beforehand given and therefore only appears
in the forcing term $\Gamma_0(t)$.
We will refer to this agent as an \emph{independent leader}.
Analyzing (\ref{example2a})-(\ref{example2b}) and assuming without loss of generality that $h_0=0$, one gathers that:
\begin{equation}
\Gamma_0(t) = \left(\begin{array}{c}
0  \\
(1-\rho)\left(fz_0(t)+g\dot z_0(t)\right) \\
0\\
\vdots \\
0\\
\rho\left(fz_0(t)+g\dot z_0(t)\right)
\end{array}\right) \quad  .
\label{eq:Gamma_0}
\end{equation}

In order to define $M$ matrix of (\ref{eq:indepleader}) in terms of these quantities,
we use the Kronecker product, $\otimes$,
$$
M= I \otimes A  + P \otimes K \; ,
$$
where $A$ and $K$ the $2\times2$ matrices:
$$
A= \left(\begin{array}{cc}
0 & 1 \\
0 & 0
\end{array}\right) \quad \mbox{ and } \quad
K= \left(\begin{array}{cc}
0 & 0 \\
f & g
\end{array}\right) \quad .
$$

The advantage of this somewhat roundabout way of defining the matrix $M$ is that in the eigenvalues of the reduced Laplacian $P$ can be given explicitly. From that the eigenvalues of $M$ can then be derived.

\section{Asymptotic Stability} \label{chap:asympt}

The system defined in (\ref{example2a})-(\ref{example2b}) is called \emph{asymptotically stable} if all eigenvalues of $M$ have negative real part. Assuming the $\Gamma_0(t)=0$, for $t>t_0$, the solution of the system tends to $0$ exponentially fast (in
$t$) if and only if the system is asymptotically stable. This corresponds to the classical notion of asymptotic stability.

The study of the eigenvalues of the $N\times N$ matrix $Q_\rho$ defined in (\ref{Q_rho})
constitutes a special case of results given in \cite{Fo1,Fo2}. They are given by:
$$2\sqrt{(1-\rho)\rho}\;\cos\left(\frac{\ell\pi}{N+1}\right)\, , \quad \mbox{for}\; \ell=1,2,\ldots,N\, , $$ for all real $\rho$.
These eigenvalues are all real if and only if $\rho\in[0,1]$ and imaginary otherwise, and
the locus of the set of eigenvalues is invariant under multiplication by $-1$. We have:

\begin{proposition}
The reduced Laplacian $P$ has eigenvalues
$\lambda_\ell=1-2\sqrt{(1-\rho)\rho}\, \cos\left(\frac{\ell\pi}{N+1}\right)$,
for $\ell=1,2,\ldots,N$, for all real values of $\rho$.
\end{proposition}

One can show that the eigenvalues of $M=I\otimes A+P\otimes K$  are the solutions $\nu_{\ell\pm}$ of the equation
\begin{equation} \label{eq:evals2}
\nu^2-\lambda_\ell\, g\, \nu-\lambda_\ell\, f= 0  \quad,
\end{equation}
where $\lambda_\ell$ runs through the spectrum of $P$
(cf. \cite{tridiagonal,flocks2,flocks4, flocks5}). So we have:

\begin{theorem} \label{theo:stable}
\begin{enumerate}
  \item The eigenvalues of $M$ are
  $$
\nu_{\ell\pm} = \frac{1}{2}\left(\lambda_\ell\, g \pm \sqrt{(\lambda_\ell\, g)^2 + 4 \lambda_\ell\, f}\right)\, ,
$$ where $\lambda_\ell$ runs through the spectrum of $P$.
  \item For $\rho \in [0,1]$, all real numbers $\lambda_\ell$ are contained in the interval $[0,2]$, and the system is asymptotically stable if and only if both $f$ and $g$ are strictly smaller than zero.
\end{enumerate}
\end{theorem}

These expressions are the same as the corresponding ones for a slghtly differnt
one dimensional flock given in \cite{flocks6}. There it was assumed that $\rho\in[0,1)$.
We extend that research by looking at real values of $\rho$ outside the interval $[0,1]$.
This may at first seem obscure. Here however is the motivation. Suppose for a moment that
one allows each agent to interact with two
neighbors on either side, then one could be tempted to model this interaction as a
discretization of a fourth derivative in the spatial variable. In that case
some of the weights of the interaction would have negative values.

\begin{theorem} \label{theo:stable2}
Let $\rho\in \mathbb{R}\backslash [0,1]$. For a given $f$ and $g$, the system defined in (\ref{example2a})-(\ref{example2b}) is asymptotically
stable, for an arbitrary $N$, if and only if both of the following hold:
\begin{enumerate}
  \item $f$ and $g$ are negative, and
  \item $f+g^2\geq 0$ or else $4|\rho(1-\rho)|\leq \frac{-g^2}{f+g^2}$.
\end{enumerate}
\end{theorem}

\begin{proof} When $\rho (1-\rho)$ is negative, the eigenvalues $\lambda_\ell$ of $P$
satisfy $\lambda_\ell = 1+i\, a_\ell$, where $a_\ell$ assumes the values
$-2\sqrt{|\rho(1-\rho)|}\;\cos \frac{\ell \pi}{N+1}$.
In particular, for an $N$ sufficiently large, the $a_\ell$'s will distribute themselves smoothly in the interval
\begin{equation}
\left( -2\sqrt{|\rho(1-\rho)|}, +2\sqrt{|\rho(1-\rho)|}\right)\, .
\label{eq:a-interval}
\end{equation}

We need to prove that eigenvalues of $M$ (provided by (\ref{eq:evals2})) have negative
real part. We first look at eigenvalues that correspond to $a_\ell$ approaching to $0$. By
continuity we may set $a_\ell=0$. We get
$$
\nu_{\pm}= \frac{g\pm \sqrt{g^2+4f}}{2} \, .
$$
These roots are real and have opposite signs if $f$ is positive and have the same
sign as $g$ if $f$ is negative. This proves that for $a_\ell$ small enough the corresponding
eigenvalues of $M$ have negative real part if and only if $f$ and $g$ are negative.

\begin{figure}[ptbh]
\centering
\includegraphics[height=3.5in]{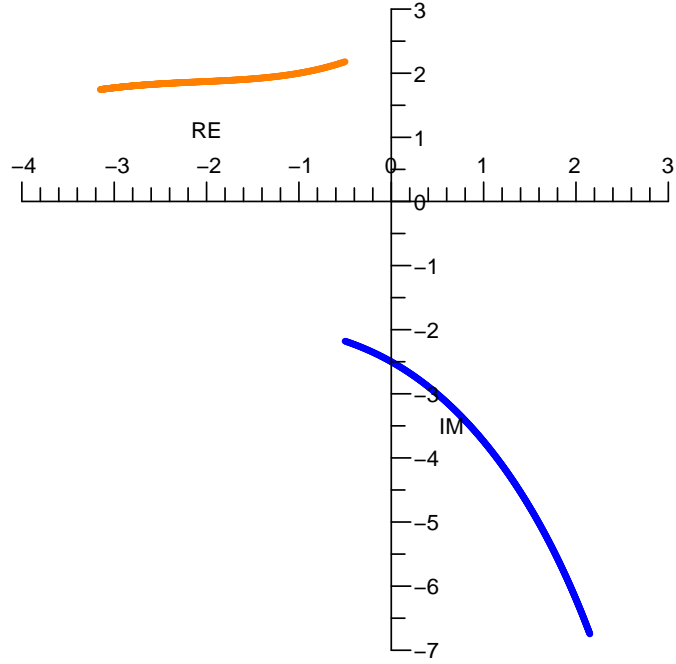}
\caption{\emph{The calculation of $\nu_\pm$ as function of $a$
for $f=-5$ and $g=-1$, $a$ ranging from 0 to 5.
Notice that $\nu_\pm(0)=0.5(-1\pm i \sqrt{19}\,)$; at $a=0.5$, $\nu_-$ crosses
the imaginary axis. }}
\label{fig:evals}
\end{figure}

It remains to check for what values of $a$ the real part of $\nu_-$ or $\nu_+$ can become
greater than or equal to $0$ (cf. Figure \ref{fig:evals}). To that end we set $\nu=i\tau$, with $\tau$ real. The real and imaginary part of
the equation (\ref{eq:evals2}) now become (abbreviating $a_\ell$ to $a$):
$$
\left\{ \begin{array}{ccc}
\tau^2 - g\, a\, \tau +f &=& 0\\
f\, a + g\,\tau   &=& 0 \end{array}
\right.
$$
The first equation defines a hyperbola in the $(a,\tau)$ plane,
and the second equation a line
through the origin. Real solutions for $a$ and $\tau$ exist if and only if $f+g^2<0$ and are given by
$$
(\tau,a)=\frac{\pm 1}{\sqrt{-f-g^2}}\, (f,-g)\, .
$$
Such solutions exist for some $a$ smaller than $2\sqrt{|\rho(1-\rho)|}$
(see interval (\ref{eq:a-interval})) if and only if in addition
$$
2\sqrt{|\rho(1-\rho)|}> \frac{- g}{\sqrt{-f-g^2}}\, .
$$

Finally, if $f+g^2<0$ and $2\sqrt{|\rho(1-\rho)|}>a> \frac{- g}{\sqrt{-f-g^2}}$,
we prove that the system has eigenvalues with positive real part. By continuity,
it is sufficient to prove this only for $a$ arbitrarily large.
In that case Theorem \ref{theo:stable} (1) gives:
\begin{equation}
\nu_-=\frac{\lambda g}{2}\;\left( 1- \sqrt{1+ \frac{4f}{\lambda g^2}} \right)
= -\frac{f}{g} + O(\lambda^{-1}) \quad .
\end{equation}
From equation (\ref{eq:evals2}) we deduce that $\mbox{Re}\, (\nu_+)+\mbox{Re}\, (\nu_-)=\mbox{Re}\, (g(1+ia))=g$.
Thus as $a$ tends to infinity, $\mbox{Re}\, (\nu_+f)=g+\frac fg >0$, which was to be proved.
\end{proof}

\section{Harmonic Stability} \label{chap:harmonic}

The system is harmonically unstable roughly if oscillatory or harmonic perturbations in the orbit
of the leader (that is: of the form $e^{i\omega t}$) have their amplitude magnified by a
factor that is exponentially large in $N$ (cf. \cite{flocks6}).

We first need some notation. It will often be convenient to replace
$\rho$ by a different constant:
$$
\kappa = \frac{1-\rho}{\rho}
$$
or, equivalently,
$$\rho = \frac{1}{1+\kappa} \, ,$$
We also define (for $\rho \neq 0$):
\begin{equation}\label{mu_pm}
\mu_\pm \equiv \frac{1}{2\rho}\left( \gamma\pm \sqrt{\gamma^2-4\rho(1-\rho)}\right)
\end{equation}
 where
$$
\gamma = \frac{f+i\,\omega\, g +\omega ^2}{f+i\,\omega\, g}\, .$$

\begin{proposition}\label{prop:a_n}
The frequency response function of the $k$-th agent is given by
$$
a_{k}(f,g,\omega)=\frac{\kappa^k(\mu_-^{N+1-k}-\mu_+^{N-k})+(\mu_-^{k}-\mu_+^{k})} {\mu_+^{N}-\mu_-^{N}}\, ,
$$
where $\mu_\pm$ is defined as in (\ref{mu_pm}).
\end{proposition}

\begin{proof} All eigenvalues of $M$ have negative real part. Let $z_0(t)$ be given by $e^{i\omega t}$. Under these assumptions, the motion of the system is asymptotic (as $t\rightarrow\infty$) to $z_k=a_k\,e^{i\omega t}$. This leads to a recursive equation on $a_k$,  for $k=1,\ldots, N$,
$$
(1-\rho)\, a_{k-1}-\gamma \, a_k +\rho \, a_{k+1}=0 \, .
$$
The boundary conditions are given by:
$$
a_0 = a_{N+1}=1  \, .
$$
Let $\mu_\pm$ be the roots of the associated characteristic polynomial
$$P(x)= x^2-\frac\gamma\rho\, x +\frac{1-\rho}{\rho}\, .$$
The general solution is $$a_k=c_-\mu_-^k+c_+\mu_+^k\, .$$ A convenient way to solve for $c_\pm$ is by setting $d_1=c_-\mu_-^N$ and $d_2=c_+\mu_+^N$. The boundary conditions can be rewritten as
$$
\left( \begin{array}{cc} 1& 1\\
 \mu_-^{N+1}& \mu_+^{N+1}\end{array} \right)
 \left( \begin{array}{c} c_-\\ c_+ \end{array} \right) =
\left( \begin{array}{c} 1\\ 1 \end{array} \right)$$ or, equivalently,
$$
 \left( \begin{array}{c} c_-\\ c_+ \end{array} \right) =
 \frac{1}{\mu_+^{N+1}-\mu_-^{N+1}} \;
  \left( \begin{array}{c} \mu_+^{N+1}-1 \\1-\mu_-^{N+1} \end{array} \right)\, .
$$
Substituting this into $a_k$ and using the fact that the product of the $\mu_\pm$ equals $\kappa$, we get the result.
\end{proof}

We will assume here without further proof that fluctuations of the leader that
are propagated through the system are largest for the agents furthest away from the
leader, \emph{i.e.} halfway in the flock. For simplicity we only consider in this section
the response of the agent $k=\frac{N+1}{2}$ where $N$ is odd (and large).

\begin{corollary} \label{cory:halfway}
If $N$ is odd, we have for $M=\frac{N+1}{2}$:
$$
a_{M}(f,g,\omega)=\frac{(\kappa^M+1)} {\mu_-^{M}+\mu_+^{M}}\, ,
$$
where $\mu_\pm$ is defined as in (\ref{mu_pm}).
\end{corollary}

Recall that by Theorems \ref{theo:stable} and \ref{theo:stable2} the system is asymptotically
stable if and only if $\rho\in[0,1]$ and $f$ and $g$ negative \emph{or else} $\rho\in \mathbb{R}\backslash [0,1]$ and both of the following hold:
\begin{itemize}
  \item $f$ and $g$ are negative, and
  \item $f+g^2\geq 0$ or else $4|\rho(1-\rho)|\leq \frac{-g^2}{f+g^2}$.
\end{itemize}

Recall from the proof of Theorem \ref{prop:a_n} that if $z_0(t)$ equals $e^{i\omega t}$ then $z_k$ is asymptotic to $a_k\,e^{i\omega t}$.
Thus the amplification at the $k$-th agent of the leader's signal is
given $a_k(\omega)$. We need to determine whether $\max_k \sup_{\omega}|a_k(\omega))|$ is exponential
in $N$ (instability) or less than exponential (stability). We may assume $k=M$. So let $$A_M\equiv \sup_{\omega\in\mathbb{R}}\; |a_M(i\omega)|\, .$$ Following \cite{flocks6}, we call a system \emph{harmonically stable} if it is asymptotically stable and if $$\limsup_{M\rightarrow \infty}\; \left|A_M\right|^{1/N}\leq 1\, .$$

\begin{theorem}
The system given by the equation (\ref{eq:indepleader}) is harmonically stable if and only if
$\rho=1/2$ and $f$ and $g$ are negative.
\label{theo:harmonic}
\end{theorem}

\begin{figure}[ptbh]
\centering
\includegraphics[height=2.3in]{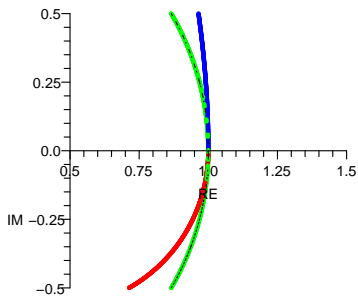}
\caption{\emph{ The eigenvalues $\mu_+(\omega)$ (blue) and $\mu_-(\omega)$ (red)
when $f=g=-1$ and $\rho=1/2$, for $\omega$ positive. In addition the unit circle is drawn in green.}}
\label{fig:rho=1/2}
\end{figure}

\begin{proof}
We concentrate first on the $\rho=1/2$ case. Here we have:
$$
a_M(\omega) = \frac{2}{\mu_+(\omega)^M + \mu_-(\omega)^M} \quad \mbox{ and } \quad \mu_+\mu_-=1 \, .
$$
Geometrically what happens is that the system exhibits near-resonance. More precisely,
the curves $\mu_\pm(\omega)$ are (quadratically) tangent to the unit circle at $\omega=0$
(see Figure \ref{fig:rho=1/2}). Of course when $\mu_\pm(\omega)=e^{\pm i\pi/2M}$ the denominator
cancels and $a_M$ is undefined. The quadratic tangency means that (for $M$ large) the curves
$\mu_\pm(\omega)$ pass the points $e^{\pm i\pi/2M}$ on the unit circle at a distance proportional
to $1/M$. In turn this means that
$$
A_M = \sup_{\omega}\, |a_M(\omega))|
$$
grows \emph{linearly} in $M$ and is thus harmonically stable.

Analytically this can be worked out precisely by doing a pole expansion on $a_M(\omega)$.
A very similar calculation was done in detail in \cite{flocks5} and we will not repeat
that calculation here. The only differences with that calculation are:
here we are calculating $a_M$ and not $a_N$, and here our eigenvalues $\nu_{\ell\pm}$ are slightly different from those in the cited paper.

Now we turn to the other cases: $\rho\neq 1/2$. We first argue that the cases $\rho$ and $1-\rho$
are symmetric. In particular if we use a new value $\rho'=1-\rho$ instead of $\rho$, then
in the expression for $a_M$, $\mu_\pm$ and $\kappa$ are all replaced by their reciprocals
as can be seen by inspecting the polynomial $P$ in the proof of Proposition \ref{prop:a_n}.
A little calculation shows that $a_M$ is invariant under this operation. It is thus sufficient
to consider only $\rho<1/2$.

Proposition \ref{prop:mu+bigger} tells us that for all $\rho<1/2$ there is no resonance
or near resonance as the two $\mu$ have distinct modulus. Lemma \ref{lem:omega+} implies that,
for $\omega$ less than some $\omega+$ given there, $\mu_-(\omega)$ is bigger than 1. Since
in this case $|\kappa|>1$ and $\mu_+\mu_-=\kappa$, we have for $\omega \in (0,\omega+)$:
$$
|\kappa|>|\mu_+(\omega)|>|\mu_-(\omega)| \quad ,
$$
and thus the expression in the above Corollary grows exponentially in $M$. Therefore all
these systems are harmonically unstable.
\end{proof}

\section{An open problem} \label{chap:impulse}

Suppose now that the flock is in a stable equilibrium (\emph{i.e.} moves stably in formation) when the
leader suddenly and quickly changes its velocity. Roughly speaking we call the system
\emph{impulse stable} when the physical response of the other agents (\emph{i.e.},
the acceleration, or the velocity, or the position) is less than exponential in $M$.
As observed in the proof of Theorem \ref{theo:harmonic}, the case $\rho=\frac12$
is extremely similar to the problem studied in \cite{flocks5}, and the solution is in fact
similar to the one given in that case. The calculations there indicate responses that
are  `proportional' to $M$. Thus may we conclude that here also:

\begin{proposition} The system given in the equation (\ref{eq:indepleader}) is impulse
stable when it is asymptotically stable and when $\rho=1/2$.
\end{proposition}

The cases $\rho\neq 1/2$ lead to problems similar to the one that remained unsolved
in \cite{flocks7}. It seems very likely at this point that most of these cases are
impulse unstable, though this is by no means obvious or known. More specifically,
as we vary $f$, $g$ and $\rho\neq 1/2$, we do not even qualitatively understand the large $N$ behavior of the motion of these flocks. This question is relevant because
velocity changes of the leader are a natural context in which stability plays
an important role for the cohesion of the flock. We might think for example of a lead car
accelerating when a traffic light turns green or a large flock of animals changing course
because outlying members spotted and try evade a predator.

The mathematical problem boils down to an inverse Fourier transform of $a_M(\omega)$
where $M$ is large. Current
standard integration techniques do not readily give asymptotic
(in $M$) expressions for such integrals. We do not address this challenging question any further here, leaving it as an open problem for future research.

\section{APPENDIX: Technical Results} \label{chap:tresults}

In this section we gather some technical results which we exhibited partially before
in \cite{flocks6}. The results there were proved only for $\rho\in[0,1]$.
Some of the calculations extend verbatim (or almost) to all $\rho\in\mathbb{R}$. The proof
of the main result, Proposition \ref{prop:mu+bigger}, had to be modified substantially however.

\begin{lemma} \label{lem:taylor} Let $\omega \geq 0$ be sufficiently small.
\begin{enumerate}
  \item For $\rho\in (0,\frac12)$, we have:
  $$
\mu_+ = \frac{1-\rho}{\rho}\left(1 + \frac{\omega^2}{(2\rho-1)|f|} - i\, \frac{|g| \,\omega^3}{(2\rho-1)f^2} \right) + \mathcal{O}(\omega^4)$$ and
$$ \mu_- =  1 - \frac{\omega^2}{(2\rho-1)|f|} + i\, \frac{|g| \, \omega^3}{(2\rho-1)f^2} + \mathcal{O}(\omega^4)  \, .$$
  \item For $\rho\in (\frac12,1)$, we have
  $$
\mu_+ = 1 - \frac{\omega^2}{(2\rho-1)|f|} + i\; \frac{|g| \, \omega^3}{(2\rho-1)f^2} + \mathcal{O}(\omega^4) $$ and
$$
\mu_- =  \frac{1-\rho}{\rho}\left(1 + \frac{\omega^2}{(2\rho-1)|f|} - i\, \frac{|g| \, \omega^3}{(2\rho-1)f^2} \right) + \mathcal{O}(\omega^4)  \, .
$$
\end{enumerate}

\end{lemma}

\begin{proof} By sheer calculation. (See \cite{flocks4} for some of the computational details.)             \end{proof}

\begin{remark} This expansion diverges for $\rho=1/2$; in that case we have (cf. \cite{flocks4}):
$$
\mu_\pm = 1-\frac{\omega^2}{|f|}\pm
\frac{\omega^2|g|}{\sqrt{2}|f|^{3/2}}+\mathcal{O}(\omega^4) +i\left( \pm
\frac{\sqrt{2}\omega}{|f|^{1/2}} \pm \mathcal{O}(\omega^3) \right)
\, .
$$
\label{remark}
\end{remark}

\begin{lemma}\label{lem:gamma} $$\gamma(\omega)=1-\frac{\omega^2 |f|}{f^2+\omega^2g^2} + i\,\frac{\omega^3|g|}{f^2+\omega^2g^2}\, .$$

\end{lemma}

The complicated looking conditions in the following proposition are
nothing but the conditions that insure asymptotic stability (see Theorems
\ref{theo:stable} and \ref{theo:stable2}).

\begin{proposition} \label{prop:mu+bigger}

Let $\rho\in[0,1]$ and $f$ and $g$ negative or let $\rho\in \mathbb{R}\backslash [0,1]$
and suppose that both of the following hold:
\begin{description}
  \item[i] $f$ and $g$ are negative, and
  \item[ii] $f+g^2\geq 0$ or else $4|\rho(1-\rho)|\leq \frac{-g^2}{f+g^2}$.
\end{description}
Then, $r$ defined by $$r=\sup_{\omega>0} \frac{|\mu_-(\omega)|}{|\mu_+(\omega)|}$$ exists and is contained in interval $(0,1)$ with only two exceptions:
\begin{itemize}
  \item when $\rho=\frac12$, $\frac{|\mu_-(\omega)|}{|\mu_+(\omega)|}$ equals $1$ at $\omega =0$ and is
is strictly smaller than $1$ for $\omega>0$, or
  \item when $4\rho(1-\rho)= \frac{-g^2}{f+g^2}<0<0$, $\frac{|\mu_-(\omega)|}{|\mu_+(\omega)|}$ is strictly
smaller than $1$, if $0\geq\omega^2<\frac{f^2}{-f-g^2}$, and $|\mu_-(\omega)|=|\mu_+(\omega)|$, if
$\omega^2=\frac{f^2}{-f-g^2}$.
\end{itemize}
\end{proposition}

\begin{figure}[ptbh]
\centering
\includegraphics[height=2.3in]{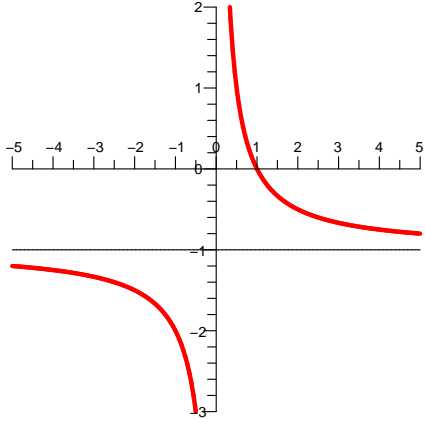}
\caption{\emph{ The constant $\kappa=(1-\rho)/\rho$ as function of $\rho$.}}
\label{fig:rho}
\end{figure}

\begin{proof} It is clear from remark \ref{remark} that when $\rho=\frac12$ the two eigenvalues are equal to $1$,
when $\omega=0$, and so at that point the quotient equals $1$.
The arguments below establish that the quotient is always smaller than $1$, for all positive values of $\omega$.
As far as the second exception is concerned: We will show that the equality of $\mu_+$ and
$\mu_-$ occurs at some $\omega>0$. The arguments below, however, insure that for smaller (non-negative)
values of $\omega$, the modulus of the quotient of the eigenvalues is strictly smaller than $1$.

From its definition, $\gamma(\omega)\approx \frac{\omega}{i g}$ when $\omega$ is large. Substitute this into the expression for $\mu_\pm$ in equation (\ref{mu_pm}) to see that for a large enough $\omega$, in fact, $\frac{|\mu_-(\omega)|}{|\mu_+(\omega)|}$ becomes very small.
In the following, note that $|\kappa|>1$ if and only if $\rho<\frac12$ and
$|\kappa|>1$ if and only if $\rho>\frac12$ (see Figure \ref{fig:rho}).
So when $\omega=0$, Lemma \ref{lem:taylor} implies that $\frac{|\mu_-|}{|\mu_+|}$ equals $\min\{|\kappa|,|\kappa|^{-1}\}$, for all $\rho$.

It is now sufficient to prove that, for $\omega\in\mathbb{R}^+$, the absolute values $|\mu_\pm|$ are never equal. So suppose there are $\omega_0$ and  $\theta\in\mathbb{R}$ such that $\mu_+(\omega_0)-\mu_-(\omega_0)e^{i\theta}=0$. The definition of $\mu_\pm$ in equation (\ref{mu_pm}) provides:
$$
\gamma (1-e^{i\theta}) = -\sqrt{\gamma^2-4\rho(1-\rho)}\,(1+e^{i\theta})\, .
$$
Dividing this by $1+e^{i\theta}$, squaring the equation, and noting that $$\frac{(1-e^{i\theta})^2}{(1+e^{i\theta})^2}=-\tan^2\left(\frac\theta2\right)\, ,$$ we get
$$
\gamma^2 \left(1+\tan^2 \frac{\theta}{2}\right)=4\rho(1-\rho) \, .
$$
If $\rho(1-\rho)>0$, then $\gamma^2$ is a positive real number and therefore $\gamma$ is real for some $\omega\neq 0$, which is impossible by Lemma \ref{lem:gamma}. If $\rho(1-\rho)<0$, then
$\gamma^2$ is a negative real number so that $\gamma$ is imaginary. Setting the real part of $\gamma$ equal
to $0$ in Lemma \ref{lem:gamma} yields $f^2+\omega^2(f+g^2)=0$. If $f+g^2$ is non-negative, this
has no solution (because $\omega$ is real). So suppose it is positive. Then substitute the positive
solution into $\gamma$ and check that $\gamma=i\frac{|g|}{\sqrt{-f-g^2}}$.
Substituting this in turn into (\ref{mu_pm}), we see that the modulus of $\mu_+$ is greater
than that of $\mu_-$ (here $\sqrt{\cdot}$ means the root in the \emph{upper} half plane)
as long as $4|\rho(1-\rho)|< \frac{-g^2}{f+g^2}$. When $4|\rho(1-\rho)|= \frac{-g^2}{f+g^2}$, we have  $\mu_+=\mu_-$.
\end{proof}

\begin{figure}[ptbh]
\centering
\includegraphics[height=2.3in]{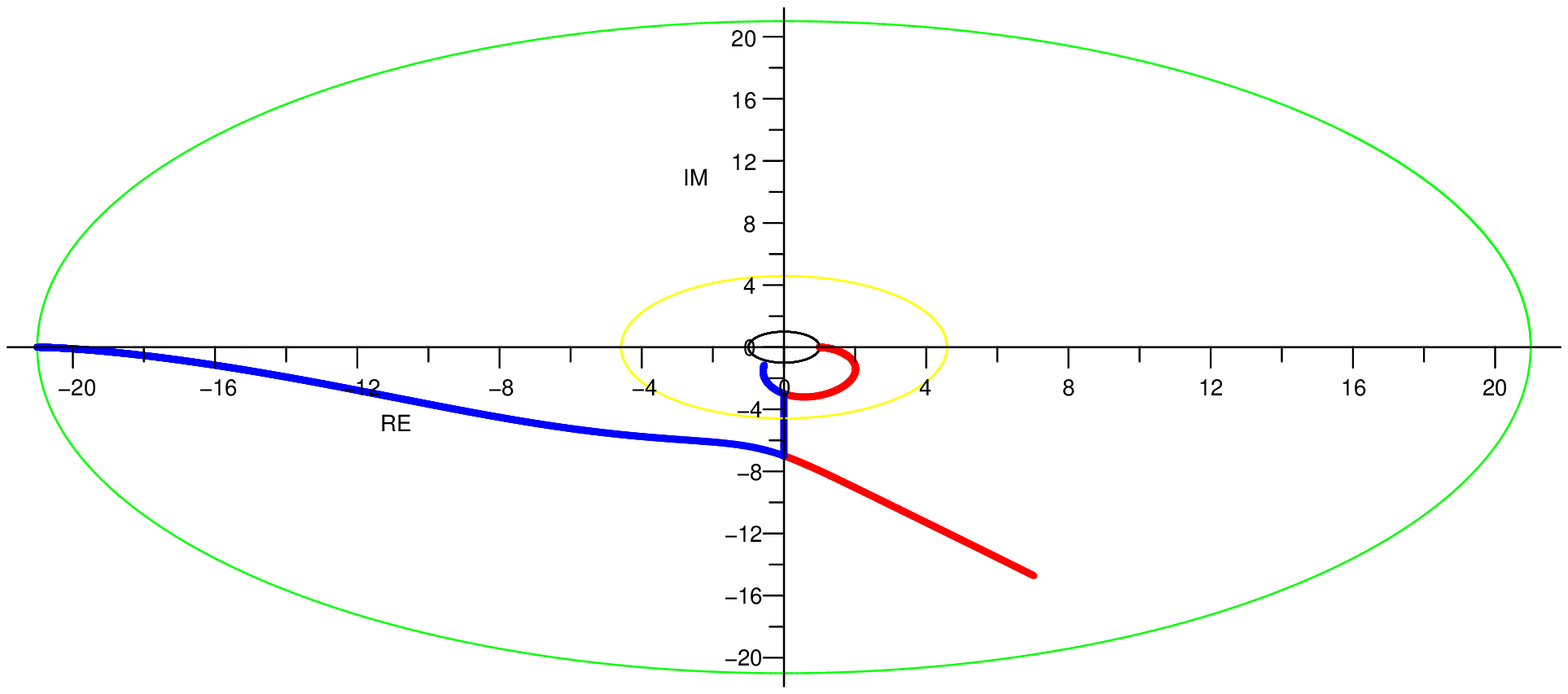}
\includegraphics[height=2.3in]{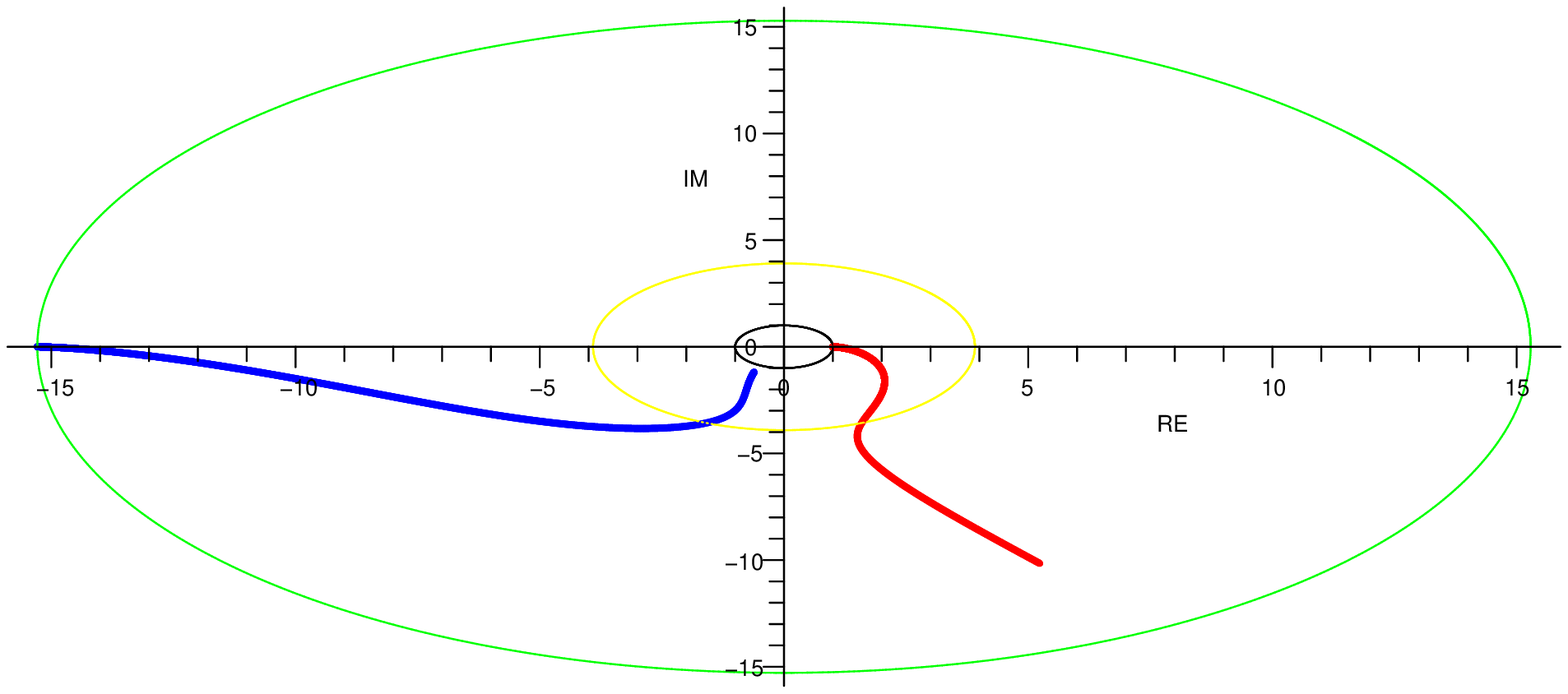}
\caption{\emph{An illustration of the curious behavior of $\mu_\pm(\omega)$.
The eigenvalues $\mu_+(\omega)$ are in blue and $\mu_-(\omega)$ are in red,
when $f=-5$ and $g=-1$, for positive $\omega$. The circle $r=\kappa$ (green),
$r=\sqrt{|\kappa|}$ (yellow), and the unit circle (black) are also drawn.
Note that $\mu_+(0)=\kappa$ (which is negative in both cases) and that $\mu_-(0)=1$.
In the first figure $\rho=-0.05$ and $\mu_{blue}$ is always bigger than $\mu_{red}$ except
MAPLE insisted in using the principal root for the square root as opposed to our convention,
and so it recklessly swaps the $2$ roots. In the second picture $4\rho(1-\rho)> \frac{-g^2}{f+g^2}<0$
(while $\rho(1-\rho) < 0$) and now the quotient of the two roots crosses $1$. }}
\label{fig:rho-extra}
\end{figure}

\begin{lemma} \label{lem:omega+} For each $\rho < 1/2$, there is a unique $\omega_+>0$ such that
$$
\omega\in\left(0,\omega_+\right)\quad \mbox{ implies } \quad |\mu_-(\omega)|>1$$
and
$$
\omega>\omega_+\quad \mbox{ implies }\quad|\mu_-(\omega)|<1 \,.
$$
\end{lemma}

\begin{proof} We know that $\mu_-(0)=1$ and, from the proof of the previous lemma, for a large $\omega$, $|\mu_-(\omega)|$ is small. It is sufficient to prove that $\omega_+$ is the unique solution in $(0,\infty)$ of $|\mu_-(\omega)|=1$ and that it is simple.

In fact, consider the characteristic equation
$\rho\, \mu^2-\gamma\, \mu + (1-\rho)=0$  and suppose that there is a root $\mu=e^{i\theta}$. Then $$\gamma=\rho e^{i\theta}+(1-\rho)e^{-i\theta} =\cos(\theta)+i(2\rho-1)\sin(\theta)\, .$$ Equating this to the expression given in Lemma \ref{lem:gamma} and using the Pythagorean trigonometric identity, we to obtain:
$$
\left(1-\frac{\omega^2 |f|} {f^2+\omega^2g^2}\right)^2+\frac{1}{(2\rho-1)^2}\left(\frac{\omega^3|g|} {f^2+\omega^2g^2}\right)^2=1
$$
This equation factors as follows:
$$
\omega^2\left(\frac{g^2}{(2\rho-1)^2}\,\omega^4+(f^2-2|f|g^2)\omega^2-2|f|^3 \right)=0
$$
The second factor is a quadratic expression in $\omega^2$ which has a positive leading
coefficient and a negative trailing coefficient. This gives exactly one simple positive root for $\omega^2$, yielding a unique simple positive root $\omega=\omega_+$.         \end{proof}

\begin{remark} In fact,
$$
\omega_+^2= (1-2\rho)^2|f|\left( 1-\frac{|f|}{2g^2}
+ \sqrt{\left(1-\frac{|f|}{2g^2}\right)^2 +\,\frac{2|f|}{(1-2\rho)^2g^2}} \right) \, .
\label{eq:omega+}
$$
\end{remark}


\begin{thebibliography}{ALS}

\bibitem{laplacians}
John S. Caughman, J. J. P. Veerman, \emph{Kernels of Directed Graph
Laplacians}, Electr. J. Comb. 13, No 1, R39, 2006.

\bibitem{Fo1}
 C.M. da Fonseca, \emph{Interlacing properties for Hermitian matrices whose graph is a given tree}, SIAM J. Matrix Anal. Appl. 27 (2005), no. 1, 130-141.

\bibitem{Fo2}
C.M. da Fonseca, \emph{On the location of the eigenvalues of Jacobi matrices}, Appl. Math. Lett. 19 (2006), no. 11, 1168-1174.

\bibitem{tridiagonal}
C.M. da Fonseca, J.J.P. Veerman, \emph{On the spectra of certain directed paths}, Appl. Math. Lett. 22 (2009), no. 9, 1351-1355.


\bibitem{flocks2}
J.J.P. Veerman, J.S. Caughman, G. Lafferriere, A. Williams, \emph{Flocks
and formations}, J. Stat. Phys. 121 (2005), Vol. 5-6, 901-936.

\bibitem{flocks4}
J.J.P. Veerman, B.D. Sto\v si\'c, A. Olvera, \emph{Spatial instabilities
and size limitations of flocks}, Netw. Heterog. Media 2 (4) (2007), 647-660.

\bibitem{flocks5}
J.J.P. Veerman, B.D. Sto\v si\'c, F.M. Tangerman, \emph{Automated traffic and the finite
size resonance}, submitted.

\bibitem{flocks6}
J.J.P. Veerman,
\emph{Stability of large flocks: an example}, submitted.

\bibitem{flocks7}
J.J.P. Veerman, F.M. Tangerman,
\emph{Impulse stability of large flocks: an example}, submitted.

\end{thebibliography}
\end{document}